\documentclass[twocolumn]{autart}    

\usepackage{cite}
\usepackage{amsmath,amssymb,amsfonts}

\usepackage{graphicx}
  
 \usepackage{array}           \usepackage{changepage}
 \usepackage{enumerate}
 \usepackage{algorithm}
 \usepackage{algpseudocode}
\newcolumntype{F}[1]{>{\raggedright\arraybackslash}p{#1}}
\newcolumntype{T}[1]{>{\centering\arraybackslash}p{#1}}
\newtheorem{theorem}{Theorem}
\newtheorem{lemma}{Lemma}
\newtheorem{remark}{Remark}
\newtheorem{definition}{Definition}
\newtheorem{assumption}{Assumption}
\newtheorem{proposition}{Proposition}

\newenvironment{proof}[1][Proof]{%
  \par\noindent\textbf{#1.}\ }{%
  \hfill$\square$\par
}

\begin{document}

\begin{frontmatter}

\title{A Systematic Approach to Mechanism Design with Stochastic Dynamic Stability}

\author{Shaya Garjani}\ead{shaya.garjani@ut.ac.ir},
\author{Mohammad Shokri}\ead{mo.shokri@ut.ac.ir},
\author{Hamed Kebriaei}\ead{kebriaei@ut.ac.ir}

\address{School of Electrical and Computer Engineering, College of Engineering, University of Tehran, Tehran, Iran}

\begin{keyword}                           
Mechanism design, resource allocation, strategic agent, Nash equilibrium.               
\end{keyword}

\begin{abstract}                          
We consider a resource allocation problem with strategic agents that have private stochastic satisfaction functions and local constraints. To achieve a global optimal solution, we propose an incentive mechanism that induces a game among the agents. For the payment function of the mechanism, we construct a family of quadratic functions using the linear matrix inequality (LMI) approach that implements the social welfare maximizing outcome on the unique Nash equilibrium (NE) of the induced game while ensuring budget balance and individual rationality. Moreover, we propose a decentralized variable sample-size proximal best-response (VS-PBR) algorithm with Krasnoselskij iteration where only aggregate information is available to the agents. The algorithm is dynamically stable, as it is proven to converge in the mean-square sense to the NE of the game. The efficiency of the mechanism is then investigated on the Sioux Falls City transportation network, where electric vehicle (EV) users jointly select their destination and route.
\end{abstract}

\end{frontmatter}

\section{Introduction}\label{sec1}
Resource allocation appears in various applications such as power distribution systems \cite{olivella2020centralised}, transportation networks \cite{bakhshayesh2021decentralized}, social networks \cite{dave2021social}, and communication systems \cite{xu2020client}. The objective in these problems is to efficiently distribute limited resources among the network's agents to maximize overall network utility.
\par
Finding the optimal network solution in resource allocation problems is quite challenging, especially when agents are selfish, prioritize privacy and act strategically. In such cases, agents aim to maximize their own utility, often conflicting with the network's objectives. Incentive mechanism design is extensively used to deal with these complexities in resource allocation \cite{chremos2024mechanism}. An incentive mechanism makes the profit-maximizing behavior of the selfish agents in line with network global optimal solution by determining communication rules for message submission, and outcome functions for allocation and payment rules. The payment function induces a game among the agents. The network manager aims to design the payment function so that agents' strategies at a Nash equilibrium (NE) point of the induced game are equivalent to the optimal solution of the resource allocation problem; this property is known as Nash implementation \cite{williams1986realization}. A more desirable feature would be strong Nash implementation, meaning that the outcome function at every NE of the game results in social welfare maximization. Other desirable properties of the mechanism include budget balance (i.e., sum of all payments at a NE is equal to zero), individual rationality (i.e., voluntary participation in the mechanism), and dynamic stability (i.e., learning algorithm provision that converges to a NE). 

Based on the communication rule, mechanisms can be categorized into direct and indirect mechanisms. One advantage of indirect mechanisms to direct ones is that agents do not have to reveal their private types to the manager \cite{narahari2014game}. Moreover, achieving strong Nash implementation often requires complex indirect mechanisms, and direct mechanisms, even in the case of the popular Vickrey-Clarke-Groves (VCG) mechanism, may result in extra inefficient equilibria \cite{rothkopf2007thirteen,JAIN20101276}.
Another shortcoming of direct mechanisms is that they often lack the budget balance property \cite{ZHONG2026112727}, \cite{satchidanandan2023incentive} or achieve weak budget balance \cite{ANGELI2023110870}, \cite{pu2020online}.

\begin{table*}[ht]

\caption{Comparison with related works.}
\centering
\begin{tabular}
{|F{0.14\textwidth}|T{0.15\textwidth}|T{0.05\textwidth}|T{0.05\textwidth}|T{0.15\textwidth}|T{0.15\textwidth}|T{0.12\textwidth}|}
\hline
\textbf{Study} & 
\textbf{Nash implementation} & 
\textbf{BB} & 
\textbf{IR} & 
\multicolumn{2}{T{0.30\textwidth}|}{\textbf{Dynamic stability}} & 
\textbf{Systematic design} \\ 
\cline{5-6}
&  &  &  & Deterministic env. & Stochastic env. &\\
\hline
\cite{ satchidanandan2023incentive, zhang2021faithful,pu2020online} & \checkmark & $\times$ & \checkmark & $\times$ & $\times$ & $\times$ \\ \hline
\cite{kakhbod2012efficient,sinha2017mechanism,heydaribeni2019distributed,rasouli2019efficient}       & \checkmark & \checkmark & \checkmark & $\times$ & $\times$ & $\times$ \\ \hline
\cite{farhadi2018surrogate, zhang2019efficient, sinha2019distributed, eslami2022incentive, zhong2022nash,yan2025social}                   & \checkmark & \checkmark & \checkmark & \checkmark & $\times$ & $\times$ \\ \hline
This paper                                       & \checkmark & \checkmark & \checkmark & 
\checkmark & \checkmark & \checkmark \\ \hline
\end{tabular}

%\end{adjustwidth}
\end{table*}

\par
There are several indirect mechanisms presented in the literature that achieve strong Nash implementation, individual rationality, and budget balance while also satisfying some additional features, such as those found in \cite{kakhbod2012efficient,sinha2017mechanism,heydaribeni2019distributed,rasouli2019efficient}. The proposed mechanism in \cite{kakhbod2012efficient} maintains a balanced budget at all Nash equilibria and off-equilibrium points. The allocation scheme for the mechanism proposed in \cite{sinha2017mechanism} results in feasible allocations even off-equilibrium. The authors of \cite{heydaribeni2019distributed} consider a message-exchange network where agents can only transmit messages to their local neighborhood, resulting in the decentralized computation of allocation and tax functions. In \cite{rasouli2019efficient}, an electricity market mechanism is proposed which, in addition to achieving the desired properties, is price-efficient and satisfies the local informational constraints of agents.

\par
Dynamic stability along with the other desirable properties of a mechanism have been obtained in a few studies, \cite{farhadi2018surrogate, zhang2019efficient, sinha2019distributed, eslami2022incentive, zhong2022nash}. In \cite{farhadi2018surrogate}, the optimal allocation is obtained via a surrogate optimization approach, and the convergence of the proposed best-estimate learning algorithm is proved based on the contraction property of the mapping. In \cite{zhang2019efficient}, an incremental subgradient-based distributed algorithm is introduced that converges to the generalized Nash equilibrium (GNE) of the game. The authors of \cite{sinha2019distributed} propose a distributed mechanism where agents communicate through a graph and the complexity of the message space grows more than linearly with respect to the number of agents. The convergence of the proposed algorithm in this paper is based on a contraction condition for the best-response mapping. The mechanism proposed in \cite{eslami2022incentive} induces a network aggregative game (NAG) among the agents, and a distributed algorithm is proposed that converges to the solutions of a corresponding variational inequality (VI) problem. The authors of \cite{zhong2022nash} provide a decentralized algorithm based on the alternating direction method of multipliers (ADMM) to reach the GNE of the game. An incentive mechanism for a pseudo-gradient-based dynamical system was proposed in \cite{yan2025social}, where sufficient conditions for its Lyapunov stability were established. The aforementioned papers do not systematically derive the payment function that achieves the desired properties. As a result, the payment functions can be too complex or only suitable for a specific application.

To the best of our knowledge, this paper is the first to systematically design the payment function of a mechanism using linear matrix inequality (LMI) optimization. We consider a resource allocation problem with private stochastic satisfaction functions and local constraints, and propose a systematic mechanism design approach to maximize social welfare under coupling network constraints. The payment function is selected from a class of quadratic functions, and its parameters are determined via an LMI framework to ensure all the desired properties of strong Nash implementation, budget balance, and individual rationality. Additionally, we develop a decentralized variable sample-size proximal best-response (VS-PBR) algorithm with Krasnoselskij iteration, which allows agents to learn their NE strategies using only aggregate information. We prove the convergence of the algorithm in the mean-square sense, ensuring stochastic dynamic stability. The main contributions of the paper are as follows:
\begin{itemize}
    \item 
    \textbf{Systematic payment design via LMI:} By parameterizing the payment function as a family of quadratic functions and imposing LMI-based conditions on its parameters, we provide a general, systematic approach to mechanism design that produces simpler payment functions and can be applied to various applications. To the best of our knowledge, this is the first such approach to mechanism design.
    \item 
    \textbf{Economic properties:} Our proposed mechanism satisfies all the desired properties like strong Nash implementation, budget balance, and individual rationality. 
    \item 
    \textbf{Decentralized learning algorithm with convergence guarantees:} We introduce a decentralized VS-PBR algorithm with Krasnoselskij iteration, enabling agents to learn their optimal NE strategies in a stochastic environment while using only aggregate information. We demonstrate that this algorithm converges to the NE of the induced game in the mean-square sense, ensuring stochastic dynamic stability.
\end{itemize}
\section{Problem Formulation}\label{problem_formulation_dec}
\subsection{Resource allocation network}\label{resource_allocation_subsec}
Consider a resource allocation problem involving $N$ selfish/strategic agents interacting in a network with $K$ resources, coordinated by a network manager. Let $\mathcal{N}=\{1, \ldots, N\}$ and $\mathcal{K}=\{1, \ldots, K\}$ denote the sets of agents and resources, respectively.
\par 
Each agent $n\in\mathcal{N}$ has a strategy vector $x_n\in\mathcal{X}_n\subset\mathbb{R}_{\geq 0}^K$, representing the amount of each resource allocated to them. 
An agent's satisfaction from consuming $x_n$ units of resources is modeled by a stochastic function 
$
\psi_n(x_n;\,\xi_n(\omega)):\ \mathcal{X}_n \times \mathbb{R}^{d_n} \rightarrow \mathbb{R}
$.
The random vector \(\xi_n \in \mathbb{R}^{d_n}\) is defined on the probability space \((\Omega, \mathcal{F}, \mathbb{P})\), where \(\mathcal{F}\) is the \(\sigma\)-algebra of events in $\Omega$ and \(\mathbb{P}\) is the corresponding probability measure. The valuation function $V_n(x_n)$,
interpreted as the average satisfaction at consumption level $x_n$, taken over all possible realizations of $\xi_n$, is defined as $V_n(x_n):= \mathbb{E}[\psi_n(x_n;\xi_n)]$, where the expectation is taken on the probability space $\mathbb{P}$. The valuation function $V_n(\cdot)$ and the strategy set $\mathcal{X}_n$ are private information of agent $n$. 
\par
The network manager aims to maximize total social welfare by solving
\begin{subequations}\label{manager_problem_eq}
	\begin{align}
	  &\max_{x\in \mathcal{X}}{\sum_{n\in\mathcal{N}}{V_n(x_n)}}\\ 
	  \text{s.t.}\quad & \sum_{n\in\mathcal{N}}{x_n}\leq c,
	\end{align}
\end{subequations}
where $x:=\textbf{col}(x_1,\dots,x_N)$, $\mathcal{X}:=\Pi_{n\in \mathcal{N}}\mathcal{X}_n$, and $c\in\mathbb{R}_{\geq 0}^K$ denotes the amount of available resources.
\begin{assumption}\label{agent_set_assumption}
For each agent $n\in\mathcal{N}$, the strategy set $\mathcal{X}_n$ is convex and compact. 
\end{assumption}
\begin{assumption}\label{valuation_function_assumption}
For each agent $n\in\mathcal{N}$, the valuation function $V_n(x_n):\mathcal{X}_n\rightarrow \mathbb{R}$ is continuous, twice differentiable, $\alpha$-strongly concave, i.e., $\forall x_n\in \mathcal{X}_n:$
\begin{equation*}\label{alpha_concave_eq}
\begin{split}
    &(\nabla V_n(x_n)-\nabla V_n(x_n'))^\top(x_n-x_n')\leq-\alpha\|x_n-x_n'\|^2,
\end{split}
\end{equation*}
and $V_n(0)=0$.
\end{assumption}
\begin{remark}
The imposition of conditions in Assumption \ref{valuation_function_assumption} on the valuation functions $V_n(x_n)$ is less strict compared to assuming the conditions for all samples of the stochastic function $\psi_n(x_n;\xi_n)$.
\end{remark}
The resource allocation problem \eqref{manager_problem_eq} with network coupling constraints cannot be solved in a centralized fashion due to the network manager's lack of knowledge about the agents' valuation functions and local constraints. In addition, agents might opt for a different strategy from the optimal solution of optimization problem \eqref{manager_problem_eq} due to their selfishness and their desire to maximize their own utility. To address these challenges, we propose an incentive mechanism.
%%%%%%%%%%%%%%%%%%%%%%%%%%%%%%%%%%%%%%
\subsection{Incentive mechanism design problem}\label{machanism_design_subsec}

\par
The network manager seeks to design an incentive mechanism that imposes a non-cooperative game among the agents, where the equilibrium of the game is the optimal solution to problem \eqref{manager_problem_eq}. To achieve this goal, the manager asks the agents to request resources and propose prices for them. Then, the manager imposes taxes or provides subsidies to the agents through payment functions, which depend on their requested resources and proposed prices, which steers the agents towards choosing the socially optimal strategies.
\par
Let $p_n\in\mathbb{R}_{\geq 0}^{K}$ denote the price per unit of resources proposed by agent $n$. We define $p:=\textbf{col}(p_1,\ldots, p_n)$ as the vector of proposed prices. The message that agent $n$ sends to the manager is defined as $s_n:=\textbf{col}(x_n, p_n)$ and its strategy set is defined by
\[\mathcal{S}_n:=\{s_n=\textbf{col}(x_n, p_n):x_n\in\mathcal{X}_n, p_n\in\mathbb{R}_{\geq 0}^{K}\}.\] 
\par
Based on the message profile $s\in\mathcal{S}$, where $s:=\textbf{col}(s_1,\dots,s_N)$ and $\mathcal{S}:=\Pi_{n\in\mathcal{N}}\mathcal{S}_n$, the manager designs the outcome function $O(s)=\{(y_n(s), t_n(s_n,s_{-n}))\}_{n\in\mathcal{N}}$ where $s_{-n}:=\textbf{col}(s_1, \dots, s_{n-1},s_{n+1}, \ldots, s_N)$ is the strategy of all agents except agent $n$, and $y_n(\cdot)$ and $t_n(\cdot,\cdot)$ denote the allocation and payment functions of the mechanism, respectively. The allocation function in the proposed mechanism is defined as $y_n(s) = x_n$, implying that the manager allocates an identical quantity of resources as requested by each agent. In the rest of the paper, we write $x_n$ instead of $y_n(s)$. The payment of the agents $t_n(s_n, s_{-n})$, introduced at the end of this section, can be zero; positive, indicating that agent $n$ has to pay taxes; or negative, implying that agent $n$ receives a subsidy. The goal of agent $n$, when carrying out strategy $s_n$ in return for payment $t_n(s_n,s_{-n})$, is to maximize its utility,
\begin{equation}\label{utility_maximization_eq}
	\max_{s_n\in\mathcal{S}_n}{[U_n(s_n, s_{-n}):=}V_n(x_n)-t_n(s_n, s_{-n})].
\end{equation}
\par
Since the payment function of each agent is affected by the strategies of other agents, a non-cooperative game is induced among the agents, which is denoted by $\mathcal{G}:=(\mathcal{N}, (\mathcal{S}_n)_{n\in\mathcal{N}}, U_n(s_n, s_{-n}))_{n\in\mathcal{N}}$. 
The Nash equilibrium (NE) of the induced game is defined as follows:

\begin{definition}[Nash Equilibrium]\label{nash_def}
	The strategy vector $s^*=\textbf{col}(s_1^*, \ldots, s_N^*)$ is a Nash equilibrium of game $\mathcal{G}$ if:
	\[
		U_n(s_n^*, s_{-n}^*)\geq U_n(s_n, s_{-n}^*), \quad \forall n\in\mathcal{N}.
	\]
\end{definition}
This implies that, at a Nash equilibrium, no agent can strictly improve its utility by unilaterally deviating from their equilibrium strategy, given that all other agents keep their strategies fixed. 

The manager aims to design the payment function $t_n$ in game $\mathcal{G}$ such that the mechanism satisfies the following desired properties:
\begin{enumerate}[P1)]
	\item \textbf{Existence and Uniqueness of NE:} There exists a unique equilibrium $s^*:=\textbf{col}(s_1^*,\ldots,s_N^*)$ for the induced game.
	\item
	\textbf{Nash implementation:} The  allocated resources at a NE of the induced game $\mathcal{G}$ are equal to the optimal solution $x^o$ of problem \eqref{manager_problem_eq}, i.e., $x_n^*=x_n^o, \, \forall n\in\mathcal{N}$.
	\item
	\textbf{Budget balance:} At the NE of the game, the total payment of agents through the whole network is equal to zero, i.e., $\sum_{n\in\mathcal{N}}{t_n(s_n^*, s_{-n}^*)}=0$.
        \item
	\textbf{Individual rationality:} All users participate in the mechanism voluntarily, in a sense that their utility at the NE of the game is at least equal to their utility when they do not participate in the mechanism, i.e., $U_n(s_n^*,s_{-n}^*)\geq V_n(0)=0, \, \forall n\in\mathcal{N}$.
	\item
	\textbf{Dynamic stability:} The manager provides an iterative algorithm, along which the agents can learn their NE strategy.
\end{enumerate}

\par
The payment function $t_n$ is chosen from a class of parameterized quadratic functions as  follows:
\begin{equation}\label{tax_function_eq}
t_n(s_n, s_{-n}):=\frac{1}{2}p^\top A^n p+p^\top B^n x+p^\top a^n.
\end{equation}
Here, $A^n$ is a symmetric matrix and the coefficient matrices and the vector are defined as:
\begin{equation}\label{tax_matrices_eq}
	A^n= \left(A_{ml}^n\right)_{m,l\in\mathcal{N}}, \, B^n= \left(B_{ml}^n\right)_{m,l\in\mathcal{N}}, \, a^n=(a_m^n)_{m\in\mathcal{N}},
\end{equation}
where $A_{ml}^n,\,B_{ml}^n\in\mathbb{R}^{K\times K}$ and $a_m^n\in\mathbb{R}^K$. The manager's objective is now to systematically determine these parameters so that the desired properties P1-P4 are satisfied.
\begin{remark}
Contrary to existing mechanisms in the literature \cite{kakhbod2012efficient, farhadi2018surrogate, eslami2022incentive}, where the structure and parameters of the payment functions are predetermined, and certain desired properties are investigated given that fully known payment function, this paper systematically constructs the payment function, starting with a parameterized quadratic function \eqref{tax_function_eq} and determining its parameters using LMI optimization method that satisfies the properties P1-P4.
\end{remark}
\begin{remark}
In mechanism design, the valuation functions of the agents are typically assumed to be unknown to the manager. This also holds in the stochastic setting considered here; thus, evaluating the expectation term is unnecessary when analyzing properties P1-P4. However, for dynamic stability (P5), during the learning process, each agent must work with a sample-based approximation of the valuation functions in order to compute its decision at each iteration.
\end{remark}
\section{Specification of the payment function}\label{sec4}
In this section, the parameters of $A^n$, $B^n$, and $a^n$ from \eqref{tax_matrices_eq} are determined through an LMI optimization framework  such that the payment functions \eqref{tax_function_eq} satisfy the properties P1-P4 stated in Section \ref{machanism_design_subsec}.
\subsection{Formulation of LMI conditions}\label{LMI_subsec}
In this section, LMI conditions are derived on the payment function parameters \eqref{tax_matrices_eq} so that the mechanism satisfies properties P1-P4.
\subsubsection{Existence and uniqueness (P1)} The induced game $\mathcal{G}$ under the designed mechanism must admit a unique equilibrium. Proposition \ref{existence_uniqueness_prop} provides LMIs that ensure the existence and uniqueness of NE.

\begin{proposition}\label{existence_uniqueness_prop}
Define the matrix $\Psi=(\Psi_{nm})_{n,m\in\mathcal{N}}$ where 
\[\Psi_{nn}=-\begin{pmatrix}
    \alpha I_K & B_{nn}^n\\
        B_{nn}^{n^\top}& A_{nn}^n
\end{pmatrix}, 
\Psi_{nm}=-\begin{pmatrix}
    0 & B_{mn}^n\\
    B_{nm}^{n^\top} & A_{nm}^n
\end{pmatrix},  m\neq n.\] A Nash equilibrium of $\mathcal{G}$ exists and is unique if
	\begin{equation}\label{existence_uniqueness_cond_eq}
	\textnormal{P1:} \left\{
	(\romannumeral 1) \Psi_{nn}\preceq 0, \,
        (\romannumeral 2) \Psi + \Psi^\top \preceq -\upsilon I_{2KN}
	\right\},
	\end{equation}
    where $\upsilon>0$ is a scalar design parameter. 
\end{proposition}
\begin{proof}
Define $J:=(J_{nm})_{n,m\in\mathcal{N}}$ where
\[
J_{nm}:=\frac{\partial^2 U_n(s_n,s_{-n})}{\partial s_n \partial s_m}.
\]
By Assumption \ref{valuation_function_assumption}, we have $\frac{\partial^2 V_n(s_n)}{\partial s_n^2}\preceq -\alpha I_K$, thus
\begin{equation}\label{unique_eq}
     J_{nn}=
     \begin{pmatrix}
        \frac{\partial^2 V_n(s_n)}{\partial s_n^2} & -B_{nn}^n\\
        -B_{nn}^{n^\top}& -A_{nn}^n
    \end{pmatrix}\preceq 
    \Psi_{nn}
\end{equation}
If $B_{nn}^n, A_{nn}^n$ are chosen such that
$\Psi_{nn}\preceq 0$, i.e., $U_n$ is concave with respect to $s_n$, there exists an equilibrium point for the game $\mathcal{G}$ \cite[Theorem 1]{rosen1965existence}.  
From \eqref{unique_eq} and the fact that for $m\neq n$, $J_{nm}=\Psi_{nm}$, it follows that $J\preceq \Psi$; therefore, 
\[\Psi+\Psi^\top\preceq -\upsilon I_{2KN} \Rightarrow J+J^\top\preceq -\upsilon I_{2KN},\]
which is a sufficient condition for the pseudogradient $\textbf{col}(\nabla_{s_1}U_1(s_1, s_{-1}), \ldots, \nabla_{s_N}U_N(s_N, s_{-N}))$ to be diagonally strictly concave \cite[Theorem 6]{rosen1965existence}, and thus, a sufficient condition for the uniqueness of NE of the game $\mathcal{G}$.
\end{proof}
\subsubsection{Nash implementation (P2)}
For P2, conditions need to be imposed on the payment rule $t_n$ such that a NE of $\mathcal{G}$ corresponds to the optimal solution of problem \eqref{manager_problem_eq}. To achieve this, it is necessary that the KKT point of \eqref{manager_problem_eq} is feasible within the KKT conditions of (\ref{utility_maximization_eq}). Theorem \ref{nash_implementation_theorem} determines the LMIs satisfying P2.
\begin{theorem}\label{nash_implementation_theorem}
	The optimal solution of problem \eqref{manager_problem_eq} is a NE of the induced game by \eqref{tax_function_eq}, if the parameters of \eqref{tax_matrices_eq} satisfy
	\begin{equation}\label{nash_implementation_cond_eq}
	\textnormal{P2:}\left\{\begin{array}{l}
		(\romannumeral 1) \sum_{m\in\mathcal{N}}{A_{nm}^n}=0, (\romannumeral 2) \sum_{m\in\mathcal{N}}{B_{mn}^n} = \text{diag}(\zeta^{n}),\\
		(\romannumeral 3) B_{nm}^n=-\text{diag}(\theta^{n}), (\romannumeral 4) a_n^{n}=\text{diag}(\theta^{n})c
		\end{array}\right\}
	\end{equation}
 where $\zeta^{n}, \, \theta^n\in\mathbb{R}_+^{K}$.
\end{theorem}

\begin{proof}
Let $\delta_{\mathcal{X}_n}(\cdot)$ denote the indicator function of the local set $\mathcal{X}_n$, i.e.,
     \[\delta_{\mathcal{X}_n}(x_n)=\begin{cases}
         0,&\text{if} \ x_n\in\mathcal{X}_n,\\
         +\infty,&\text{otherwise}.
     \end{cases}
     \]
 The KKT conditions of problem \eqref{manager_problem_eq} are
		\begin{subequations}\label{eq10}
			\begin{align}
				& \nabla_{x_n} V_n(x_n^o)-\lambda^o-\partial_{x_n} \delta_{\mathcal{X}_n}(x_n^o)=0, \quad \forall n\in\mathcal{N}\\
			  & \lambda^{o\top}\bigg(\sum_{n\in\mathcal{N}}{x_n^{o}}-c\bigg)=0,\\
                & \sum_{n\in\mathcal{N}}{x_n^{o}}-c^{}\leq 0, \,\lambda^{o}\geq 0
			\end{align}
		\end{subequations}
		where $x_n^o$ and $\lambda^o$ are the optimal primal and dual variables of problem \eqref{manager_problem_eq} for $n\in \mathcal{N}$, and $\partial_{x_n}$ denotes the subgradient w.r.t. $x_n$. For each agent $n\in\mathcal{N}$, the KKT conditions of (\ref{utility_maximization_eq}) are
		\begin{subequations}\label{eq11}
			\begin{align}
			\begin{split}
			&\nabla_{s_n} V_n(x_n^*) -\nabla_{s_n}t_n(s_n^*, s_{-n}^*)-\partial_{s_n} \delta_{\mathcal{X}_n}(x_n^*)\\
			&+\nabla_{s_n}(\mu_n^{*\top}p_n^*)=0, \quad \forall n\in\mathcal{N},\\
			\end{split}\\
                & \mu_n^{*\top}p_n^{*}=0, \quad \forall n\in\mathcal{N},\\
                & p_n^{*}\geq 0, \, \mu_n^{*}\geq 0, \quad  \forall n\in\mathcal{N}
		  \end{align}
		\end{subequations}
		where $x_n^*$ and $p_n^*$ are the optimal solution of (\ref{utility_maximization_eq}) given the optimal solutions of the other agents $s_{-n}^*$, and $\mu_n^*$ is the dual variable of $p_n^*\geq 0$. 
 \par
 Let $s^*=\textbf{col}(x^*, p^*)$ denote the NE strategies of agents that satisfy KKT conditions \eqref{eq11}. We will show that by substituting $x^*$ with $x^o$ in (\ref{eq11}a) and by imposing conditions on the parameters \eqref{tax_matrices_eq}, $p^*$ can be found such that $s^*=\textbf{col}(x^o, p^*)$ satisfies KKT conditions \eqref{eq11}. In what follows, conditions on the payment functions are derived such that the optimal solution of \eqref{manager_problem_eq} is implemented at a NE with symmetric pricing,  
 i.e.,  
 $p_n^*=p_{m}^*=\tilde{p}, \, \forall n,{m}\in\mathcal{N}$.
 
 By substituting $x^o$ and $\tilde{p}$ for all $n\in\mathcal{N}$ with their counterparts in \eqref{eq11}, and taking the partial derivation in equation (\ref{eq11}a)  w.r.t. $x_n$ and $p_n$, we obtain
		\begin{subequations}\label{eq12}
			\begin{align}
				&\nabla_{x_n}V_n(x_n^o)=\sum_{m\in\mathcal{N}}{B_{mn}^{n^\top}\tilde{p}}+\partial_{x_n} \delta_{\mathcal{X}_n}(x_n^o)\\
				&\sum_{m\in\mathcal{N}}{A_{nm}^{n}\tilde{p}}+\sum_{m\in\mathcal{N}}{B_{nm}^nx_m^o}+a_n^{n}-\mu_n^*=0
			\end{align}		
		\end{subequations}
Comparing (\ref{eq12}a) and (\ref{eq10}a) we derive
    \begin{equation}\label{eq13}
    \sum_{m\in\mathcal{N}}{B_{mn}^{n^\top}\tilde{p}}=\lambda^o, \quad \forall n\in\mathcal{N}.
    \end{equation}
    By requiring $\sum_{m\in\mathcal{N}}{B_{mn}^{n^\top}}\succ 0$, $\tilde{p}$ can be calculated by 
  \begin{equation}\label{eq14}
      \tilde{p}=(\sum_{m\in\mathcal{N}}{B_{mn}^{n^\top}})^{-1}\lambda^o,
  \end{equation}
  and since $\lambda^o\geq 0$, for $p_n^{*}=\tilde{p}\geq 0$ in (\ref{eq11}c) to be satisfied, \[\sum_{m\in\mathcal{N}}{B_{mn}^{n^\top}}\in\mathbb{R}_{\geq 0}^{K\times K}\] must hold.
		
  Now, by applying (\ref{nash_implementation_cond_eq}$\romannumeral 1$) to (\ref{eq12}b), we obtain \[\sum_{m\in\mathcal{N}}{B_{nm}^n}x_n^o+a_n^{n}=\mu_n^*.\]
  By choosing $B_{nm}^n=-\tilde{B}^n$ and $a_n^n=\tilde{B}^n c$ where $\tilde{B}^n\in\mathbb{R}_{\geq 0}^{K\times K}$, we have
  \begin{equation}\label{eq16}
      -\tilde{B}^n(\sum_{n\in\mathcal{N}}{x_n^o}-c)=\mu_n^*,
  \end{equation} 
  which satisfies $\mu_n^{*}\geq 0$ in (\ref{eq11}c). 

 Considering \eqref{eq14} and \eqref{eq16}, (\ref{eq11}b) is equivalent to \[-\lambda^{o\top}(\sum_{m\in\mathcal{N}}{B_{mn}^{n}})^{-1}\tilde{B}^n(\sum_{n\in\mathcal{N}}{x_n^o}-c)=0,\] which can be rewritten as
\begin{equation}\label{eq_17}
\begin{aligned}
    \begin{split}
        \lambda^{o\top}\Lambda_\Phi(\sum_{n\in\mathcal{N}}{x_n^o}-c) + \lambda^{o\top}(\Phi-\Lambda_\Phi)(\sum_{n\in\mathcal{N}}{x_n^o}-c)=0
    \end{split}
\end{aligned}
\end{equation}
where, $\Phi=(\sum_{m\in\mathcal{N}}{B_{mn}^{n}})^{-1}\tilde{B}^n$ and $\Lambda_\Phi$ is the diagonal part of matrix $\Phi$. From (\ref{eq10}b) and (\ref{eq10}c), it is evident that \[\lambda^{o(k)}(\sum_{n\in\mathcal{N}}{x_n^{o(k)}}-c^{(k)})=0,\quad \forall k\in\mathcal{K},\] meaning that the first term in \eqref{eq_17} is equal to zero. Considering (\ref{eq10}c) and the fact that all entries of the matrix $\Phi$ are non-negative we have

\[\lambda^{o\top}(\Phi-\Lambda_\Phi)(\sum_{n\in\mathcal{N}}{x_n^o}-c)\le 0,\] and the equality is attained only with 
$\Phi = \Lambda_\Phi$. Therefore, under conditions (\ref{nash_implementation_cond_eq}$\romannumeral 2$)-(\ref{nash_implementation_cond_eq}$\romannumeral 4$), the KKT conditions (\ref{eq11}b) and (\ref{eq11}c) are satisfied. 

Thus, by applying  \eqref{nash_implementation_cond_eq} to the parameters of the payment functions, a NE strategy for the agents is $s_n^*=\textbf{col}(x_n^o, \tilde{p})$ for all $n\in\mathcal{N}$. This completes the proof.
\end{proof}
From P1 and P2 one can deduce that the optimal solution of problem (1) is implemented on the unique NE of $\mathcal G$. Thus, the mechanism achieves strong Nash implementation.
\subsubsection{Budget balance (P3)}
Theorem \ref{theorem3} provides conditions that result in the budget balance (P3) property of the mechanism.
\begin{theorem}\label{theorem3}
	If the parameters of \eqref{tax_matrices_eq} satisfy
	\begin{equation}\label{eq19}
		\textnormal{P3:} \left\{\begin{array}{l}
		 (\romannumeral 1) \sum_{m\in\mathcal{N}\backslash n}{A_{nn}^m}=A_{nn}^n, \sum_{m\in\mathcal{N}\backslash n}{B_{nn}^m}=B_{nn}^n\\
		 (\romannumeral 2) A_{ml}^n=0, \, B_{ml}^n=0, \, m\neq l, \ m,l\in \mathcal{N}\backslash n,\\
		 (\romannumeral 3) \sum_{m\in\mathcal{N}\backslash n}{a_n^m}=-\frac{1}{N}\sum_{m\in\mathcal{N}}{B_{nm}^m}c
		\end{array}\right\},
	\end{equation}
	the mechanism is budget balanced, which implies that $\sum_{n\in\mathcal{N}}{t_n(s_n^*, s_{-n}^*)}=0$.
\end{theorem}
\begin{proof}
	The budget balance property is investigated at the NE point $s_n^*=\textbf{col}(x_n^o, \tilde{p}), \, \forall n\in\mathcal{N}$, where
	\begin{equation}\label{eq20}
	\begin{aligned}
	\begin{split}
	t_n(s_n^*, s_{-n}^*)&=\frac{1}{2}\tilde{p}^{\top}\sum_{m\in\mathcal{N}}{\left[\sum_{l\in\mathcal{N}}{(A_{ml}^n\tilde{p}+2B_{ml}^nx_l^o)}+2a_m^n\right]}\\
	\end{split}
	\end{aligned}.	
	\end{equation}
	From (\ref{eq11}b) and (\ref{eq12}b), it follows that
	\begin{equation}\label{eq21}
		\tilde{p}^\top\left(\sum_{m\in\mathcal{N}}{A_{nm}^{n}\tilde{p}}+\sum_{m\in\mathcal{N}}{B_{nm}^nx_m^o}+a_n^{n}\right)=0.
	\end{equation}
	By adding and subtracting $\frac{1}{2}\tilde{p}^{\top}A_{nn}^n\tilde{p}$ to the right-hand side of (\ref{eq20}) and using \eqref{eq21}, we have
	\begin{equation}\label{eq22}
	\begin{aligned}
	\begin{split}
	t_n(s_n^*, s_{-n}^*)&=-\frac{1}{2}\tilde{p}^{\top}A_{nn}^n\tilde{p}+\frac{1}{2}\tilde{p}^{\top}\sum_{m\in\mathcal{N}\backslash n}{\sum_{l\in\mathcal{N}\backslash n}{A_{ml}^n}}\tilde{p}\\&
	+\tilde{p}^{\top}\sum_{m\in\mathcal{N}\backslash n}{\sum_{l\in\mathcal{N}}{B_{ml}^nx_l^o}}+\tilde{p}^{\top}\sum_{m\in\mathcal{N}\backslash n}{a_m^n}\\
	\end{split}
	\end{aligned}	
	\end{equation}
	From (\ref{eq13}), we derive that \[x_n^{o\top}(\sum_{m\in\mathcal{N}}{B_{mn}^{n^\top}\tilde{p}})= x_n^{o\top}\lambda^o.\] By adding and subtracting $\tilde{p}^{\top}B_{nn}^nx_n^o$ to \eqref{eq22}, we obtain
	\begin{equation}\label{eq23}
	\begin{aligned}
	\begin{split}
	t_n(s_n^*, &s_{-n}^*)=-\frac{1}{2}\tilde{p}^{\top}A_{nn}^n\tilde{p}-\tilde{p}^{\top}B_{nn}^nx_n^o+\lambda^{o\top}x_n^o\\&+\frac{1}{2}\tilde{p}^{\top}\sum_{m\in\mathcal{N}\backslash n}{\sum_{l\in\mathcal{N}\backslash n}{A_{ml}^n}}\tilde{p}\\&+\tilde{p}^{\top}\sum_{m\in\mathcal{N}\backslash n}{\sum_{l\in\mathcal{N}\backslash n}{B_{ml}^nx_l^o}}+\tilde{p}^{\top}\sum_{m\in\mathcal{N}\backslash n}{a_m^n}\\
	\end{split}
	\end{aligned}	
	\end{equation}
    To remove the quadratic terms w.r.t. $s_n$ in \eqref{eq23} and the terms $\tilde{p}^{\top} A_{ml}^n\tilde{p}$ and $\tilde{p}^{\top} B_{ml}^nx_l^o$, the following must hold:
    \begin{equation}\label{eq24}
    \begin{split}
            \sum_{m\in\mathcal{N}\backslash n}{A_{nn}^m}=A_{nn}^n, \,
            \sum_{m\in\mathcal{N}\backslash n}{B_{nn}^m}=B_{nn}^n,\\ 
            A_{ml}^n=0, \, B_{ml}^n=0, \quad m\neq l, \ m,l\in \mathcal{N}\backslash n.
    \end{split}
    \end{equation}
    Summing (\ref{eq23}) over $n\in\mathcal{N}$, and using \eqref{eq24}, we obtain
	\[
		\sum_{n\in\mathcal{N}}{t_n(s_n^*, s_{-n}^*)}=\lambda^{o\top}\sum_{n\in\mathcal{N}}{x_n^{*}}+\tilde{p}^{\top}\sum_{n\in\mathcal{N}}{\sum_{m\in\mathcal{N}\backslash n}{a_n^m}},
	\]
	and using the fact that at the NE, (\ref{eq10}b) holds, we have
	\begin{equation}\label{eq27}
	\sum_{n\in\mathcal{N}}{t_n(s_n^*, s_{-n}^*)}=\lambda^{o\top}c+\tilde{p}^{\top}\sum_{n\in\mathcal{N}}{\sum_{m\in\mathcal{N}\backslash n}{a_n^m}}
	\end{equation}
	By (\ref{eq13}), we derive that $\sum_{n\in\mathcal{N}}{\sum_{m\in\mathcal{N}}{B_{nm}^{m^\top}}\tilde{p}} = N\lambda^o$. Using this equation, a condition for Equation \eqref{eq27} to be zero is (\ref{eq19}$\romannumeral 3$). 
\end{proof}
%%%%%%%%%%%%%%%%%%%%%%%
\subsubsection{Individual rationality (P4)} The following proposition gives the sufficient conditions for individual rationality of the mechanism, that is $U_n(s_n^*,s_{-n}^*)\geq V_n(0)=0$.
\begin{proposition}\label{IR_prop}
    The mechanism is individually rational if the parameters of \eqref{tax_matrices_eq} satisfy
    \begin{equation}\label{IR_cond_eq}
        \textnormal{P4:} \left\{\begin{array}{l}
		 (\romannumeral 1) \sum_{m\in\mathcal{N}}{\sum_{l\in\mathcal{N}}{A_{ml}^n}}\preceq 0, (\romannumeral 2) \sum_{m\in\mathcal{N}}{a_{m}^n}\leq 0,\\
		 (\romannumeral 3) -B_{ml}^n\in\mathbb{R}_{\geq 0}^{K\times K}, \, m\in\mathcal{N}, \ l\in \mathcal{N}\backslash n\\
		\end{array}\right\}.
    \end{equation}
    
\end{proposition}
\begin{proof}
    Let $s_n'=\textbf{col}(0, \tilde{p})$. We will impose conditions on parameters \eqref{tax_matrices_eq} such that $U_n(s_n', s_{-n}^*)\geq 0$. Using the fact that $V_n(0)=0$, it is sufficient to show that 
    \begin{equation*}
    \begin{aligned}
        \begin{split}
            t_n(s_n', s_{-n}^*)&=\frac{1}{2}\tilde{p}^\top \sum_{m\in\mathcal{N}}{\sum_{l\in\mathcal{N}}{A_{ml}^n}}\tilde{p} \\&+ \tilde{p}^\top\sum_{m\in\mathcal{N}}{\sum_{l\in\mathcal{N}\backslash n}{B_{ml}^n x_l^o}} + \tilde{p}^\top\sum_{m\in\mathcal{N}}{a_m^n}
        \end{split}
    \end{aligned}
    \end{equation*}
    is non-positive. The quadratic term w.r.t. $\tilde{p}$ has only non-positive values if (\ref{IR_cond_eq}$\romannumeral 1$) is satisfied. Since $\tilde{p}$ and $x_l^o, \, \forall l\in\mathcal{N}$, are non-negative vectors, for the other terms to be non-positive, all entries of the vector $\sum_{m\in\mathcal{N}}{a_m^n}$ and the matrices $B_{ml}^n$ must be non-positive, i.e., (\ref{IR_cond_eq}$\romannumeral 2$) and (\ref{IR_cond_eq}$\romannumeral 3$) must hold. Using the definition of NE, we have \begin{equation}U_n(s_n^*, s_{-n}^*)\geq U_n(s_n',s_{-n}^*)\geq 0,\end{equation} and therefore, the mechanism has the individual rationality property.
\end{proof}

\subsection{A feasible solution to linear matrix inequalities}
In this section, a feasible solution satisfying the LMI conditions of Propositions \ref{existence_uniqueness_prop}-\ref{IR_prop} and Theorems \ref{nash_implementation_theorem}-\ref{theorem3} from Section \ref{LMI_subsec} is found; hence, a quadratic payment function is designed that results in the desired properties P1-P4. The following proposition introduces such a payment function.
\begin{proposition}\label{prop2}
	The payment function in \eqref{eq40} satisfies the LMIs \eqref{existence_uniqueness_cond_eq}, \eqref{nash_implementation_cond_eq}, \eqref{eq19}, and \eqref{IR_cond_eq}, and hence results in a mechanism satisfying properties P1–P4.
	\begin{equation}\label{eq40}
	\begin{aligned}
	\begin{split}
	t_n(&s_n, s_{-n})=\alpha\bigg[\frac{1}{2}p_n^\top(p_n-\frac{2}{N-1}\bar{p}_{-n})\\&-\frac{1}{N-1}p_n^\top(x_n+\bar{x}_{-n}-c)+\frac{N}{(N-1)^2}\bar{p}_{-n}^\top(x_n-\frac{c}{N})\\&+\frac{1}{2(N-1)}\bar{\sigma}_{-n}^p-\frac{1}{(N-1)^2}(\bar{\sigma}_{-n}^{px}-\bar{p}_{-n}^\top\frac{c}{N})\bigg]\\
	\end{split}
	\end{aligned}
	\end{equation}
	where 
    \begin{align*}
        &\bar{p}_{-n}=\sum_{m\in\mathcal{N}\backslash n}{p_m}, \quad \bar{x}_{-n}=\sum_{m\in\mathcal{N}\backslash n}{x_m},\\
        &\bar{\sigma}_{-n}^p=\sum_{m\in\mathcal{N}\backslash n}{p_m^\top p_m},\quad \bar{\sigma}_{-n}^{px}=\sum_{m\in\mathcal{N}\backslash n}{p_m^\top x_m}.
    \end{align*}
\end{proposition}
\begin{proof}
	The proof can be straightforwardly deduced by directly applying conditions \eqref{existence_uniqueness_cond_eq}, (\ref{nash_implementation_cond_eq}), (\ref{eq19}), and (\ref{IR_cond_eq}) to (\ref{eq40}).
\end{proof}
In (\ref{eq40}), the quadratic term with respect to $s_n$ penalizes deviation from uniform pricing. The second term shows the payment corresponding to the coupling constraint, meaning if agents violate the limitations of resources, they are fined, and if they use less than the maximum amount of the resources, they get paid. The expression $\bar{p}_{-n}^\top(x_n-\frac{c}{N})$  depicts the payment of agent $n$ based on the aggregate proposed price of others. The additional components, which are quadratic terms with respect to $s_{-n}$, are needed to satisfy the LMI conditions. At the equilibrium point, this payment function results in a uniform price which is equal to $\lambda^o/\alpha$, the Lagrange multiplier associated with each constraint in \eqref{manager_problem_eq}, and also yields feasible demands that satisfy the coupling constraints. 
\section{Stochastic Dynamic Stability}
In this section, the dynamic stability (P5) of the induced game is investigated and a decentralized variable sample-size proximal best-response (VS-PBR) \cite{lei2022distributed} algorithm with Krasnoselskij iteration is proposed, through which the agents learn their NE strategy with partial knowledge about their rivals' decisions. 
\par
By eliminating terms in \eqref{eq40} that do not depend on the strategy of user $n\in\mathcal{N}$, we can define \[\hat{t}_n(s_n, \bar{s}_{-n}):=\alpha [s_n^\top Ms_n + (\Delta \bar{s}_{-n}+d)^\top s_n]\] where $\bar{s}_{-n}=\sum_{m\in \mathcal{N}\backslash n}{s_m}$, $d=\textbf{col}(0, \frac{c}{N-1})$, and 
\[
            M=-\begin{pmatrix}
                0 & 0\\
                \frac{1}{N-1}I_{K} & \frac{1}{2}I_{K}
            \end{pmatrix},  \Delta = \frac{-1}{N-1}\begin{pmatrix}
                0 & I_{K}\\
                \frac{-N}{(N-1)}I_{K} & I_{K}
            \end{pmatrix}.
\]

Problem \eqref{utility_maximization_eq}, is equivalent to the following problem:
\begin{equation}\label{u_hat_def_eq}
    \max_{s_n\in\mathcal{S}_n}{[\widehat{U}_n(s_n, \bar{s}_{-n})}:=V_n(x_n)-\hat{t}_n(s_n, \bar{s}_{-n})].
\end{equation}
The proximal best-response for \eqref{u_hat_def_eq} at iteration $i$ for each agent $n\in\mathcal{N}$ is defined as
\[
T_n(s^i) := \arg\min_{z_n\in\mathcal{S}_n}{\left[-\widehat{U}_n(z_n, \bar{s}_{-n}^i)+\frac{\mu}{2}\|z_n-s_n^i\|^2\right]},
\]
and the proximal best-response mapping of the whole game can be denoted as $T(s^i)=\textbf{col}(T_1(s^i), \ldots, T_N(s^i))$.
\begin{assumption}\label{ass2}
	For each $n\in\mathcal{N}$: \begin{enumerate}[(i)]
	    \item The gradient of the valuation function $\nabla_{x_n}V_n(x_n)$ is $L_{V_n}$-Lipschitz continuous in $\mathcal{X}_n$, i.e., \[\|\nabla_{x_n}V_n(x_n)-\nabla_{x_n}V_n(x_n')\|\leq L_{V_n}\|x_n-x_n'\| \forall x_n\in\mathcal{X}_n.\]
        \item For all $\xi_n\in\mathbb{R}^{d_n}$, $\psi_n(x_n;\xi_n)$ is differentiable in $x_n\in\mathcal{X}_n$ and there exists $\kappa^n>0$ such that \[\mathbb{E}[\|\nabla_{x_n}V_n(x_n)-\nabla_{x_n}\psi_n(x_n;\xi_n)\|^2]\leq \kappa^{n^2}.\]
        
	\end{enumerate}  
\end{assumption}
\begin{assumption}\label{ass3}
	The agents choose their pricing strategy from a compact set $\{p_n\in\mathbb{R}^K_+:p_n^k\leq P_{\text{max}}^k, k\in\mathcal{K}\}$.
\end{assumption}
\begin{remark}
	Based on Assumptions \ref{valuation_function_assumption} and \ref{ass2}, along with the KKT conditions of \eqref{manager_problem_eq}, ${p_n^*}^{(k)}$ is bounded. Consequently, the network manager can set a sufficiently large $P_{\text{max}}^k$ to ensure that the NE is not restricted. 
\end{remark}
\begin{assumption}\label{ass5}
    The proximal best-response mapping $T(s)$ is non-expansive on the compact set $s\in\mathcal{S}$, i.e., \[\|T(s)-T(s')\|\leq\|s-s'\|, \, \forall s,s'\in\mathcal{S}.\]
\end{assumption}
\begin{remark}
   The non-expansiveness in Assumption \ref{ass5} can be shown for quadratic valuation functions \cite{grammatico2015decentralized}. For the special case of $V_n(x_n)=\frac{\alpha}{2}\|x_n\|^2+b^\top x_n$, using the Gershgorin Circle theorem \cite[Theorem 1.1]{varga2011gervsgorin}, $T(\cdot)$ is non-expansive if $\mu\geq \frac{N}{N-1}\alpha$.
\end{remark}
Supposing Assumptions \ref{ass3} and \ref{ass5}, the Krasnoselskij iteration given by 
\begin{equation}\label{krasnoselskij_def}
    F(s^i) := (1-\tau)s^i +\tau T(s^i),\quad \tau\in (0,1),
\end{equation} 
converges to a fixed point of $T(\cdot)$ \cite[Theorem 3.2]{berinde2007iterative}.
\begin{algorithm}[!t]
	\caption{Decentralized VS-PBR with Krasnoselskij Step.}
	\begin{algorithmic}
		\State \textbf{Initialization: } $i\leftarrow0$, $\tau\in(0,1)$, $\mu>0$, and $s_n(0)\in \mathbb{R}_{\geq 0}^{2K}$ for $n \in  {\mathcal{N}}$
		%%%%%%%%%%%%%%%%%%%%%%%%%%%%%%%%%%
		\State \textbf{Iteration} $i$
		\vspace{0.05cm}
		\State\hspace*{0.02\linewidth}\vline
		\vspace{0.05cm} 
		\begin{minipage}{0.98\linewidth}
			%%%%%%%%%%%%%%%%%%%%%%%%%%%%%%%%%%
                \State   \textbf{Manager}: 
			\State\hspace*{0.02\linewidth} \vline
			\begin{minipage}{0.9\linewidth}
				\State\hspace*{0.02\linewidth}$
                \bar{s}^i=\sum_{n\in\mathcal{N}}{s^i}$
			\end{minipage}
			\State   \textbf{Agent}: $n \in  \mathcal{N}$
			\State\hspace*{0.02\linewidth} \vline
			\begin{minipage}{0.9\linewidth}
				\State\hspace*{0.02\linewidth}$\widehat{T}_n(s^i)=\arg\min\limits_{s_n\in\mathcal{S}_n}{}\big[-\frac{1}{Q_i}\sum_{q=1}^{Q_i}{\psi_n(x_n; \xi_{n,i}^q)}$
                \State\hspace*{0.16\linewidth}$+\hat{t}_n(s_n, \bar{s}^i - s_n^i)+\frac{\mu}{2}\|s_n-s_n^i\|^2\big]$
				\State\hspace*{0.01\linewidth} $s_n^{i+1}=(1-\tau)s_n^i+\tau \widehat{T}_n(s^i)$
			\end{minipage}
                %%%%%%%%%%
			\vspace{0.1cm}
			%%%%%%%%%%%%%%%%%%%
			\State ${i}\leftarrow {i}+1$
		\end{minipage}
	\end{algorithmic}
	\label{algorithm}
\end{algorithm}
\par
Now, we analyze the convergence properties of the sample-average PBR problem under the Krasnoselskij iteration in Algorithm \ref{algorithm}. It is assumed that at iteration $i$, we have $Q_i$ i.i.d. realizations $\{\xi_{n,i}^q\}_{q=1}^{Q_i}$ of the random vector $\xi_n$. To demonstrate the convergence of the algorithm, it suffices to show that $\lim\limits_{i\rightarrow\infty}{\|s^i-T(s^i)\|}=0$. First, we determine the bound on the inexactness of the approximated solution at each iteration of algorithm \ref{algorithm}, denoted by $\varepsilon_n^{i+1}:=E[\|s_n^{i+1}-F_n(s^i)\|^2]$.
\begin{lemma}\label{lemma2}
	Suppose Assumption \ref{ass2} holds. Define $L_{\widehat{U}_n}:=\sqrt{L_{V_n}^2+\|A_{nn}^n\|^2+2\|B_{nn}^n\|^2}$, and \[C_b^n:=\frac{\mu}{\sqrt{\mu^2+L_{\widehat{U}_n}^2}\left(\sqrt{\mu^2+L_{\widehat{U}_n}^2}-L_{\widehat{U}_n}\right)}.\] Then, for all $n\in\mathcal{N}$, it follows that $\varepsilon_n^{i+1}\leq \frac{\tau^2\kappa^{n^2}C_b^{n^2}}{Q_i}$.
\end{lemma}
\begin{proof}
	From Assumption \ref{ass2} we know that in the compact set $\mathcal{X}_n$, $\nabla_{x_n}V_n(x_n)$ is $L_{V_n}$-Lipschitz continuous. Moreover, the map $\nabla_{s_n} \hat{t}_n$ is Lipschitz continuous in the compact set $\mathcal{S}_n$ with the constant $L_{\hat{t}_n}:=\sqrt{\|A_{nn}^n\|^2+2\|B_{nn}^n\|^2}$. Therefore, using the definition of $\widehat{U}_n$ in \eqref{u_hat_def_eq}, we have
    \[\|\nabla_{s_n}\widehat{U}_n(s_n, \bar{s}_{-n}^i)-\nabla_{s_n}\widehat{U}_n(s_n', \bar{s}_{-n}^i)\|\leq L_{\widehat{U}_n}\|s_n-s_n'\|.\]
    From $s_n^{i+1}-F_n(s^i)=\tau(\widehat{T}_n(s^i)-T_n(s^i))$ and following the proof of \cite[Lemma 8]{lei2022distributed}, we have $\varepsilon_n^{i+1}\leq \frac{\tau^2\kappa^{n^2}C_b^{n^2}}{Q_i}$.
\end{proof}
\begin{theorem}\label{theorem4}
	Suppose Assumptions \ref{ass3} and \ref{ass5}, and Lemma \ref{lemma2} hold. Define $C_b:=\max_{n}{[\kappa^{n^2}C_b^{n^2}]}$ and $Q_i:=\left\lceil\frac{C_b}{\eta^{2(i+1)}}\right\rceil$ for some $\eta\in(0,1)$. Then, Algorithm \ref{algorithm} converges to the NE of game $\mathcal{G}$, in the mean-square sense. 
\end{theorem}
\begin{proof}
	Let $s^*$ be a fixed point of $T$. From \[s_n^{i+1} - s_n^* = s_n^{i+1} - F_n(s^i) + F_n(s^i) - s_n^*,\] and using the definition of $F_n(s^i)$ in \eqref{krasnoselskij_def}, it can be deduced that
	\[
		s_n^{i+1}-s_n^*=s_n^{i+1}-F_n(s^i)+(1-\tau)(s_n^i-s_n^*)+\tau(T_n(s^i)-s_n^*).
	\]
	Taking the squared norm of this equation, we have
	\begin{equation}\label{eq33}
		\begin{aligned}
		\begin{split}
		\|&s_n^{i+1}-s_n^*\|^2=\|s_n^{i+1}-F_n(s^i)\|^2+(1-\tau)^2\|s_n^i-s_n^*\|^2\\
		&+\tau^2\|T_n(s^i)-s_n^*\|^2+2(1-\tau)\langle s_n^i-s_n^*, s_n^{i+1}-F_n(s^i)\rangle\\
		%&\\
		&+2\tau\langle T_n(s^i)-s_n^*, s_n^{i+1}-F_n(s^i)\rangle\\
		&+2\tau(1-\tau)\langle s_n^i-s_n^*,T_n(s^i)-s_n^*\rangle.
		\end{split}
		\end{aligned}
	\end{equation}
	In addition, for any constant $a$, we have \[a(s_n^i-T_n(s^i))=a(s_n^i-s_n^*)-a(T_n(s^i)-s_n^*),\] and therefore,
	\begin{equation}\label{eq34}
	\begin{aligned}
	\begin{split}
	a^2\|s_n^i-T_n(s^i)\|^2=&a^2\|s_n^i-s_n^*\|^2+a^2\|T_n(s^i)-s_n^*\|^2\\
	& -2a^2\langle s_n^i-s_n^*, T_n(s^i)-s_n^*\rangle.
	\end{split}
	\end{aligned}
	\end{equation}
	Since $T_n(s^*)=s_n^*$, by Assumption \ref{ass5}, we obtain
    \[\|T_n(s^i)-s_n^*\|=\|T_n(s^i)-T_n(s^*)\|\leq\|s_n^i-s_n^*\|.\]
     Adding (\ref{eq33}) and (\ref{eq34}), and using the above inequality, we have
	\begin{equation}\label{eq35}
		\begin{aligned}
		\begin{split}
		\|&s_n^{i+1}-s_n^*\|^2+a^2\|s_n^i-T_n(s^i)\|^2\\
		&\leq (2a^2+\tau^2+(1-\tau)^2)\|s_n^i-s_n^*\|^2+\|s_n^{i+1}-F_n(s^i)\|^2\\
		&\quad+2(\tau(1-\tau)-a^2)\langle s_n^i-s_n^*,T_n(s^i)-s_n^*\rangle\\
		&\quad+2(1-\tau)\langle s_n^i-s_n^*, s_n^{i+1}-F_n(s^i)\rangle\\
		&\quad+ 2\tau\langle T_n(s^i)-s_n^*, s_n^{i+1}-F_n(s^i)\rangle.
		\end{split}
		\end{aligned}
	\end{equation}
	Using the Cauchy-Schwarz inequality, and choosing $a$ such that $a^2\leq \tau(1-\tau)$, we have 
    \begin{equation*}
    \begin{aligned}
         &(\tau(1-\tau)-a^2)\langle s_n^i-s_n^*, T_n(s^i)-s_n^*\rangle\leq\\& (\tau(1-\tau)-a^2)\|s_n^i-s_n^*\|\|T_n(s^i)-s_n^*\|\leq \\&(\tau(1-\tau)-a^2)\|s_n^i-s_n^*\|^2.
    \end{aligned}
    \end{equation*}
    Similarly, we can write 
    \begin{equation*}\langle s_n^i-s_n^*, s_n^{i+1}-F_n(s^i)\rangle\leq \|s_n^i-s_n^*\|\|s_n^{i+1}-F_n(s^i)\|,\end{equation*} 
    and 
    \begin{equation*}
    \begin{aligned}
         \langle T_n(s^i)-s_n^*, s_n^{i+1}&-F_n(s^i)\rangle \leq \\&\|T_n(s^i)-s_n^*\|\|s_n^{i+1}-F_n(s^i)\|\\&=\|T_n(s^i)-T_n(s^*)\|\|s_n^{i+1}-F_n(s^i)\|\\&\leq \|s_n^i-s_n^*\|\|s_n^{i+1}-F_n(s^i)\|.
    \end{aligned}  
	\end{equation*} 
    Substituting these bounds into the right-hand side of inequality \eqref{eq35}, we obtain
	\begin{equation}\label{eq36}
	\begin{aligned}
		\begin{split}
		\|s_n^{i+1}-s_n^*\|^2&+a^2\|s_n^i-T_n(s^i)\|^2\\
		&\leq\|s_n^i-s_n^*\|^2+ \|s_n^{i+1}-F_n(s^i)\|^2\\
		& +\|s_n^i-s_n^*\|\|s_n^{i+1}-F_n(s^i)\|.
		\end{split}
	\end{aligned}
	\end{equation}
	Moving $\|s_n^{i+1}-s_n^*\|^2$ to the right-hand side of (\ref{eq36}) and summing over iterations $i=\{0,\ldots, I\}$, we have
	\begin{equation}\label{eq37}
		\begin{aligned}
			\begin{split}
			&a^2\sum_{i=0}^{I}{\|s_n^i-T_n(s^i)\|^2}\leq \sum_{i=0}^{I}{\left[\|s_n^i-s_n^*\|^2-\|s_n^{i+1}-s_n^*\|^2\right]}\\
			&+\sum_{i=0}^{I}{\|s_n^{i+1}-F_n(s^i)\|^2}+\sum_{i=0}^{I}{\|s_n^i-s_n^*\|\|s_n^{i+1}-F_n(s^i)\|}\\
            &=\|s_n^0-s_n^*\|^2-\|s_n^{I+1}-s_n^*\|^2
			+\sum_{i=0}^{I}{\|s_n^{i+1}-F_n(s^i)\|^2}\\&+\sum_{i=0}^{I}{\|s_n^i-s_n^*\|\|s_n^{i+1}-F_n(s^i)\|}.
			\end{split}
		\end{aligned}
	\end{equation}
	Since $\mathcal{S}_n$ is a compact set, we can define $C^{n,i}$ such that $\mathbb{E}[\|s_n^i-s_n^*\|^2]\leq C^{n,i}$ and denote $C^{n,\text{max}}:=\max_i{C^{n,i}}$. Using Lemma~\ref{lemma2} and  $Q_i=\left\lceil\frac{\max_{n}{[\kappa^{n^2}C_b^{n^2}]}}{\eta^{2(i+1)}}\right\rceil\geq \frac{\kappa^{n^2}C_b^{n^2}}{\eta^{2(i+1)}}$, we can show that \[\mathbb{E}[\|s_n^{i+1}-F_n(s^i)\|^2]\leq \tau^2\eta^{2(i+1)}.\] By taking expectations on \eqref{eq37} and using H{\"o}lder's inequality $\mathbb{E}[\|XY\|]\leq\sqrt{\mathbb{E}[\|X\|^2]\mathbb{E}[\|Y\|^2]}$, we obtain that 
 \begin{equation*}
 \begin{aligned}
     \begin{split}
         a^2\sum_{i=0}^{I}{\mathbb{E}[\|s_n^i-T_n(s^i)\|^2]}&\leq C^{n,0}+\tau^2\sum_{i=0}^{I}{\eta^{2(i+1)}}\\&+\tau\sqrt{C^{n,\text{max}}}\sum_{i=0}^{I}{\eta^{(i+1)}}.
     \end{split}
 \end{aligned}
 \end{equation*}
 Denote $C^0=\max_n{[C^{n,0}]}$ and $C^{\text{max}}=\max_n{[C^{n,\text{max}}]}$, and choose $a^2=\tau(1-\tau)$. Since $\eta\in (0,1)$, as $I\to\infty$, we have
    \begin{equation*}
        \begin{aligned}
           \sum_{i=0}^{\infty}&{\mathbb{E}[\|s^i-T(s^i)\|^2]}\leq\\ &\frac{N}{\tau(1-\tau)}\left(C^0+\frac{\tau\eta[\tau\eta+\sqrt{\frac{C^{\text{max}}}{N}}(1+\eta)]}{1-\eta^2}\right)<\infty.
        \end{aligned}
    \end{equation*}
    Hence, the sequence $\{\mathbb{E}[\|s^i-T(s^i)\|^2]\}_{i\ge0}$ is summable. Since all terms are nonnegative, it follows from standard results on nonnegative summable sequences \cite[Theorem 3.24]{rudin2021principles} that
$
\mathbb{E}[\|s^i-T(s^i)\|^2] \to 0
$ as $i\to\infty$,
establishing mean-square convergence to a fixed point $s^*=T(s^*)$.  Since the game $\mathcal{G}$ admits a unique Nash equilibrium on the compact set $\mathcal{S}$, and fixed points of $T$ correspond to Nash equilibria, it follows that Algorithm \ref{algorithm} converges to $s^*$ from any initial condition $s^0 \in \mathcal{S}$.
\end{proof}
Based on Theorem \ref{theorem4}, we observe that the algorithm provided by the manager converges to the NE of game $\mathcal{G}$, in the presence of random disturbances. Therefore, the mechanism with the payment function \eqref{eq40} satisfies the stochastic dynamic stability property (P5), together with properties P1-P4.
\section{Numerical Simulations}
In this section, we evaluate the effectiveness of the proposed mechanism through simulations on a transportation network inspired by \cite{bakhshayesh2021decentralized}, with $N$ strategic electric vehicle (EV) users. The network is modeled as a directed graph $G=(\mathcal{V},\mathcal{E})$, where $\mathcal{V}={1,\ldots,V}$ is the set of nodes, including road intersections, charging stations $(\mathcal{H}\subseteq\mathcal{V})$, and user origins ($o_n\in\mathcal{V}$). The edge set $\mathcal{E}\subseteq\mathcal{V}\times\mathcal{V}$ represents the road segments.

Each EV user has a charging demand $q_n$ and satisfies it by simultaneously selecting a charging station and a route to it, following the strategy $x_n:=\textbf{col}(r_n,d_n)$. Here, $r_n:=\textbf{col}(r_n^1,\ldots,r_n^E)$ denotes the probability distribution over routes, while $d_n:=\textbf{col}(d_n^1,\ldots,d_n^H)$ specifies the allocation of charging demand across stations. The valuation function of user $n$ is defined as
\begin{equation*}
\begin{aligned}
\begin{split}
    V_n(x_n) &:= -\mathbb{E}\bigg[\frac{\alpha_n}{2}(\|x_n-\Tilde{x}_n\|^2-\|\tilde{x}_n\|^2)\\&+w_n\sum_{e\in\mathcal{E}}{(a_e+\xi_{a,n}^e)r_n^e}+\sum_{h\in\mathcal{H}}{(\rho_h+\xi_{\rho,n}^h)d_n^h}\bigg],
\end{split}
\end{aligned}
\end{equation*}
where $\Tilde{x}_n$ is the user’s preferred strategy, $w_n$ is the weight associated with travel time, $a_e$ is the average travel time on edge $e$, and $\rho_h$ is the electricity price at station $h$. The zero-mean disturbances $\xi_{a,n}^e$ and $\xi_{\rho,n}^h$ capture the randomness due to weather and grid conditions. The first term reflects satisfaction from choosing preferred routes and stations, the second term represents perceived travel costs, and the last term accounts for charging expenses.
The local constraints for each driver $n\in\mathcal{N}$ are defined as
\begin{equation*}
    \mathcal{X}_n:=\{x_n|r_n^e\in [0,1], \ d_n^h\in \mathbb{R}_{\geq 0}^H, \sum_{h\in\mathcal{H}}{d_n^h}\leq q_n, \ \mathcal{A}r_n=b_n\},
\end{equation*} 
where $\mathcal{A}\in \{0,1,-1\}^{V\times E}$ is the node-edge incidence matrix with entries $\mathcal{A}_{ve}=1$ if edge $e$ enters node $v$, $\mathcal{A}_{ve}=-1$ if edge $e$ leaves node $v$, and $\mathcal{A}_{ve}=0$ otherwise.
Here, $b_n:=\textbf{col}(b_n^1, \ldots, b_n^V)$ where for each $v\in\mathcal{V}$,
\begin{equation*}
    b_n^v=\begin{cases}
    -\frac{1}{q_n}\sum_{h\in\mathcal{H}}{d_n^h},&\text{if}\ v=o_n,\\
    d_n^v/q_n,&\text{if}\ v\in \mathcal{H},\\
    0,& \text{otherwise}.
\end{cases} 
\end{equation*} 
The constraint $\mathcal{A}r_n=b_n$ captures the coupling between route selection and charging decisions. In particular, it enforces flow conservation at the origin and charging stations, implying that each user departs from its origin with total probability $\frac{1}{q_n}\sum_{h\in\mathcal{H}}{d_n^h}$, and distributes this probability across charging stations proportionally to $d_n^h/q_n$. There are also network coupling constraints, 
$       \sum_{n\in\mathcal{N}}{r_n^e} \leq c_e^r, \
        \sum_{n\in\mathcal{N}}{d_n^h} \leq c_h^d,
$
representing the limited capacity on roads and energy capacity at stations.

\begin{table}[t!]
	\caption{Charging Stations' Energy Capacity and Electricity Price.}
	\centering
	\setlength{\tabcolsep}{4pt}
	\begin{tabular}{|l|c|c|c|c|c|c|}
		\hline
		\textbf{Station [k]} & \textbf{1} & \textbf{2} & \textbf{3} & \textbf{4} & \textbf{5} & \textbf{6}\\
		\hline
		\textbf{$\pmb{c_k^d}$ (kWh)} & 481.5 & 259.2 & 518.5 & 592.6 & 444.4 & 333.3\\
		\hline
		\textbf{$\pmb{\rho_k}$ (\textcent/kWh)} & 11.3 & 8.4 & 9.2 & 12.4 & 10.7 & 8.6\\
		\hline
	\end{tabular}\label{table}	
\end{table}

\begin{figure}[t!]
	\centering
	\includegraphics[width=8cm,height=10cm]{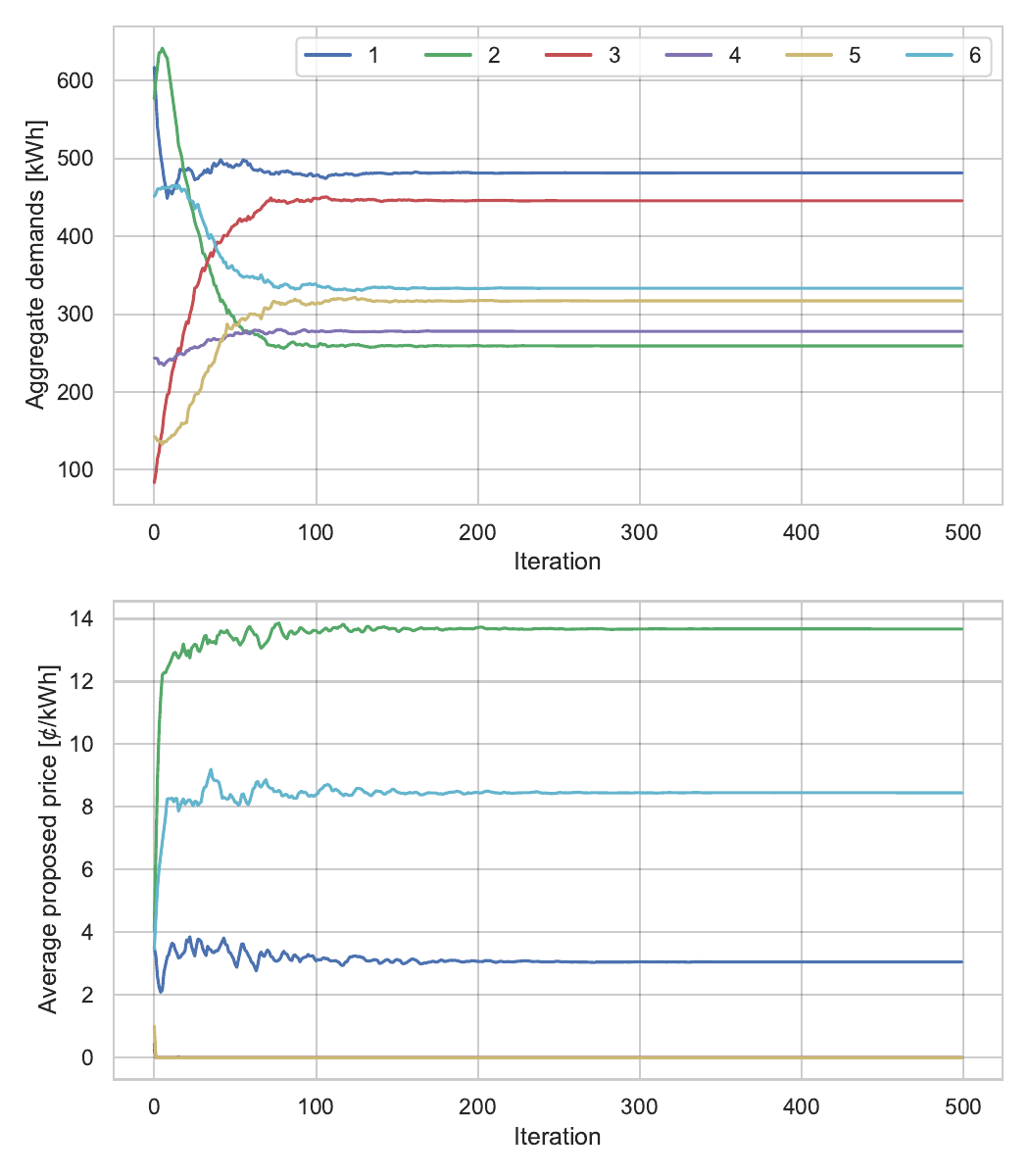}
	\caption{Dynamic stability (P5) of Algorithm \ref{algorithm}. Trajectories of the aggregate charging demand, $\sum_{n\in\mathcal{N}}{d_n^h}$ (top), and average proposed unit price, $\frac{1}{N}\sum_{n\in\mathcal{N}}{p_n^h}$ (bottom), for each charging station $h\in\mathcal{H}$.}
	\label{dynamic_stability_fig}
\end{figure}

Our case study is based on the Sioux Falls City transportation system \cite{chakirov2014enriched}, which consists of $24$ nodes and $76$ edges, of which $6$ nodes are designated as charging stations. The free-flow travel time on each edge is computed as $a_e=\frac{\text{length}}{\text{FFS}}$, where the edge length is obtained from Google Maps and the free-flow speed (FFS) is drawn from a uniform distribution $FFS\sim U(40,70)$ km/h. Road capacity is modeled as $c^r_e = 4 \times \text{FFS}$, assuming a traffic density of four vehicles per kilometer. The random perturbation in travel time is given by $\xi_{a,n}^e\sim U(-a_e/5,a_e/5)$. For charging stations, the available energy and electricity prices are listed in Table \ref{table}. The price perturbation at station $h$ is modeled as $\xi_{\rho,n}^h\sim U(-\rho_h/5,\rho_h/5)$. We simulate $50$ EV users with randomly assigned origins, electricity demands $q_n \sim U(20,70)$ kWh, monetary time value taken from a uniform distribution $w_n \sim U(10,40)$ \textcent/h, and $\alpha_n=0.8$.

\par
We run Algorithm \ref{algorithm} with $\tau=0.2$, $\mu=1$, and $Q_i=\lceil \tilde{\eta}^{i+1}\rceil$ where $\tilde{\eta}=0.96$. Figure \ref{dynamic_stability_fig} illustrates the dynamic stability (P5) of Algorithm \ref{algorithm}, as it shows the convergence of aggregate charging demands $\sum_{n\in\mathcal{N}}{d_n^h}$ for $h\in\mathcal{H}$ in the top panel, and the convergence of average unit price $\tfrac{1}{N}\sum_{n\in\mathcal{N}}{p_n^h}$ for $h\in\mathcal{H}$ in the bottom panel.

 Figure \ref{nash_implement_fig} illustrates the Nash implementation property (P2) of the proposed mechanism through the relative error with respect to the optimal solution of the centralized optimization problem \eqref{manager_problem_eq}, given by $\tfrac{\|x^i-x^o\|}{\|x^0-x^o\|}$, over the iterations. The top panel shows the evolution of the relative error for different sampling degrees $\tilde{\eta}\in\{0.96,0.97,0.98\}$. For $\tilde{\eta}=0.96$, the relative error decreases to $10^{-7}$ within 400 iterations. The bottom panel illustrates the evolution of the relative error for different Krasnoselskij coefficients $\tau \in \{0.2,0.5,0.8\}$. It can be seen that the relative error decreases faster for $\tau = 0.2$. This is because, in the Krasnoselskij iteration, smaller values of $\tau$ assign greater weight to the previous iterate, thereby damping stochastic fluctuations and yielding smoother and more stable convergence trajectories.

 \begin{figure}[t!]
	\centering
	\includegraphics[width=8cm,height=10cm]{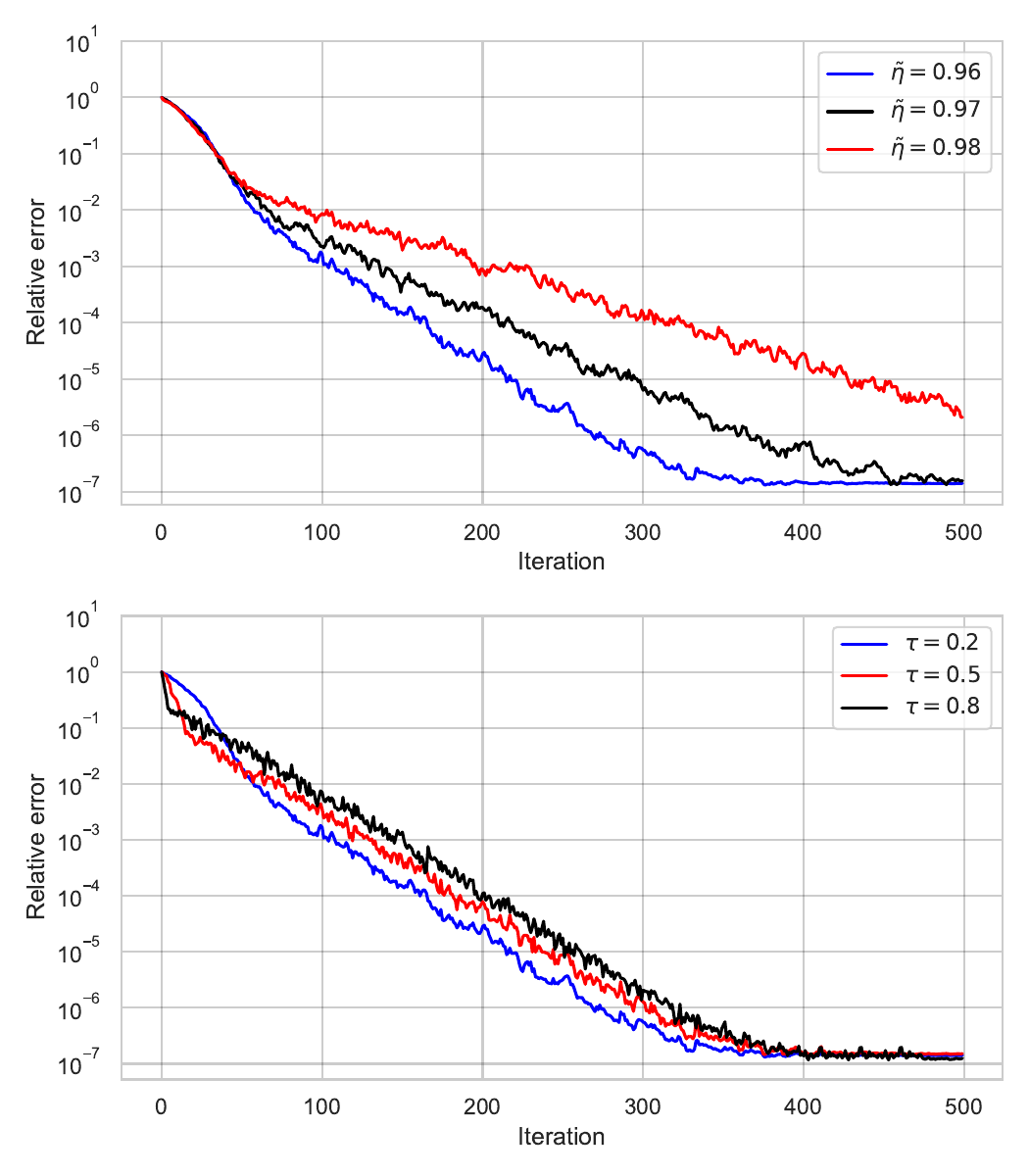}
	\caption{Nash implementation property (P2). Relative error 
$\tfrac{\|x^i-x^o\|}{\|x^0-x^o\|}$ 
with respect to the sampling degree $\tilde{\eta}$ (top) and the Krasnoselskij coefficient $\tau$ (bottom).}
	\label{nash_implement_fig}
\end{figure}

The budget balance property (P3) of the mechanism is illustrated in Figure \ref{budget_balance_fig}. The figure shows that the aggregate payments of all users $\sum_{n\in\mathcal{N}}{t_n(s_n,s_{-n})}$ rapidly converge to zero. This behavior is induced by the quadratic term in the payment functions \eqref{tax_function_eq}, which penalizes deviations from uniform pricing. Figure \ref{individual_rationality_fig} shows the evolution of users' utilities $U_n(s_n,s_{-n})$ through iterations. We can see that at the equilibrium, all utilities are non-negative, implying that the individual rationality property (P4) of the mechanism is satisfied.

\begin{figure}[t!]
	\centering
	\includegraphics[width=8cm,height=5cm]{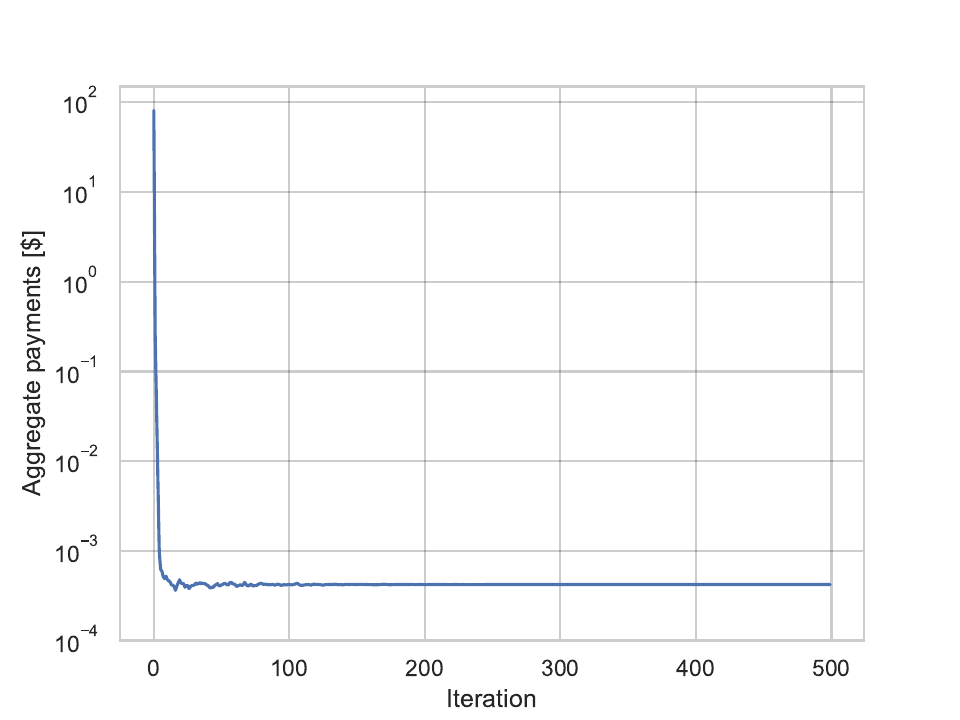}
	\caption{Budget balance property (P3). Evolution of the aggregate payment $\sum_{n\in\mathcal{N}}{t_n(s_n,s_{-n})}$ over the iterations.}
	\label{budget_balance_fig}
\end{figure}

\begin{figure}[t!]
	\centering
	\includegraphics[width=8cm,height=5cm]{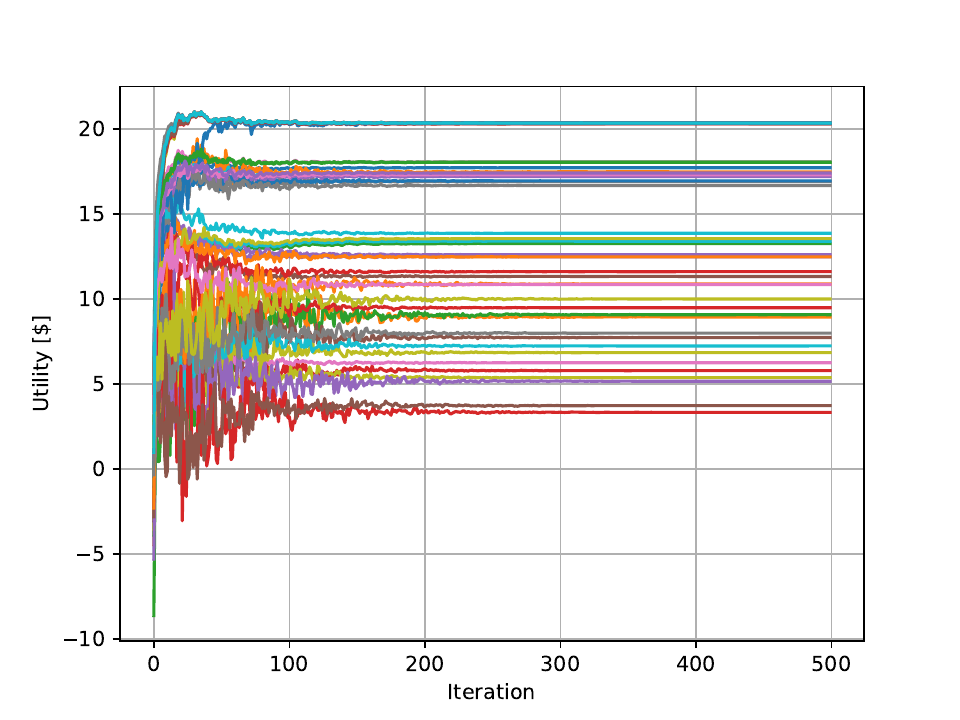}
	\caption{Individual rationality property (P4). Evolution of the users' utilities $U_n(s_n,s_{-n})$ over the iterations.}
	\label{individual_rationality_fig}
\end{figure}

Figure \ref{stations_img} demonstrates EV users’ charging behavior with and without the imposition of payment functions. As can be seen, the payment mechanism effectively enforces the coupling constraints on users’ strategies and also determines the average taxing price for charging at each station. When available energy at a station is scarce (or equivalently, when demand is high), service prices increase, e.g., at stations 2 and 6. Conversely, when resources are abundant (or demand is low), prices decrease (e.g., station 1).
\begin{figure}[t!]
	\centering
	\includegraphics[width=8cm,height=5cm]{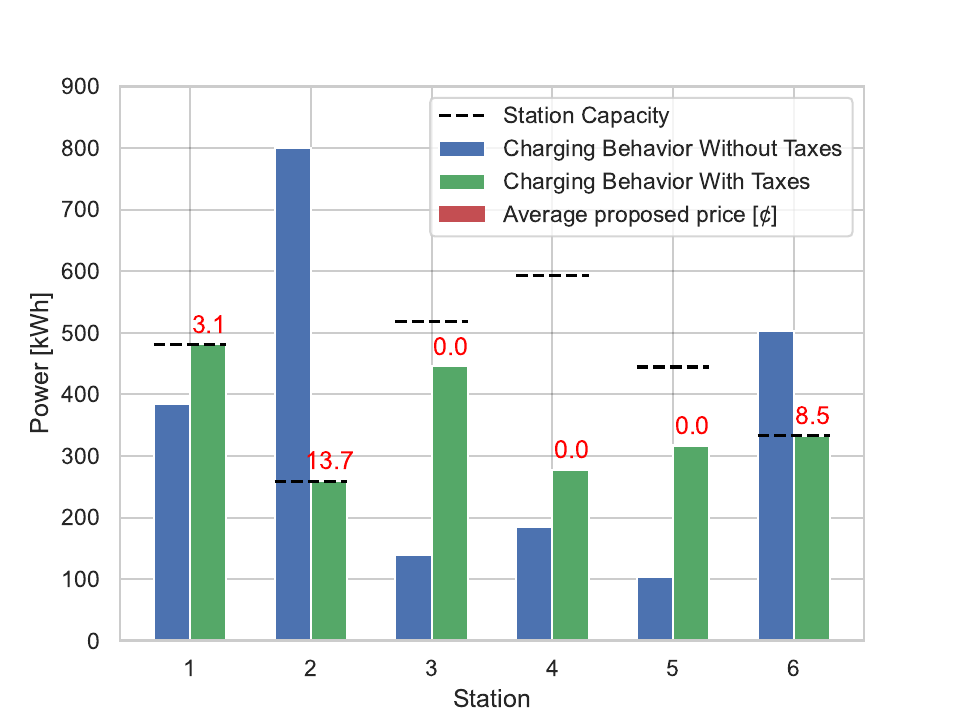}
	\caption{EV users charging behavior and average proposed taxing price.}
	\label{stations_img}
\end{figure}
\section{Conclusion}\label{sec7}
In this paper, an incentive mechanism was proposed to induce a game among agents that implements the solution of a centralized resource allocation problem. First, the payment function of the mechanism was designed systematically by imposing LMI conditions and solving the optimization. Then, a sample-based learning algorithm was provided and its mean-square convergence to the NE was investigated.

As a future work, one can propose an optimization approach which obtains a mechanism among the agents who connect with each other through the graph. Also, dynamic mechanism design can be considered as another future research direction.

\bibliographystyle{plain}        
\bibliography{autosam}

\end{document}